\documentclass[a4paper,10pt]{article}
\usepackage[utf8]{inputenc}
\usepackage{amssymb,amsthm,amsmath}
\usepackage{fullpage}
\title{On the relative power of reduction notions in arithmetic circuit complexity}
\author{Christian Ikenmeyer\\Max Planck Institute for Informatics, Saarland Informatics Campus, Germany \and Stefan Mengel\\CNRS, CRIL UMR 8188, France}

\numberwithin{equation}{section}
\swapnumbers

\newtheorem{lemma}[equation]{Lemma}
\newtheorem{theorem}[equation]{Theorem}
\newtheorem{corollary}[equation]{Corollary}

\newcommand{\HC}{\mathrm{HC}}
\newcommand{\vnp}{\mathsf{VNP}}	
\newcommand{\vp}{\mathsf{VP}}	
\newcommand{\sP}{\mathsf{\#P}}

\newcommand{\FF}{\mathbb{F}}
\newcommand{\con}{\mathrm{con}}
\newcommand{\Sp}{\textup{p}}
\newcommand{\Sc}{\textup{c}}

\begin{document}

\maketitle

\begin{abstract}
We show that the two main reduction notions in arithmetic circuit complexity, \Sp-projections and \Sc-reductions, differ in power. We do so by showing unconditionally that there are polynomials that are $\vnp$-complete under \Sc-reductions but not under \Sp-projections. We also show that the question of which polynomials are $\vnp$-complete under which type of reductions depends on the underlying field.
\end{abstract}

{\footnotesize
\\[1ex]
\\[1ex]
\noindent\textbf{2012 ACM Computing Classification System:} Theory of computation -- Computational complexity and cryptography -- Problems, reductions and completeness
}

\section{Introduction}
While there is a plethora of different reduction notions that have been studied in computational complexity (see e.g.~\cite{hemaspaandraO02} for an overview), it has often been observed that in nearly all $\mathsf{NP}$-completeness proofs in the literature logarithmic space many-one reductions suffice. 
In contrast to $\mathsf{NP}$-completeness, many $\sP$-hardness results in counting complexity are not shown with many-one reductions but with the more permissive Turing-reductions and the $\sP$-hardness under many-one reductions remains an open problem.
It is natural to ask if there is a fundamental difference between both $\sP$-hardness notions.
Note that the question of the relative power of reduction notions for $\mathsf{NP}$-completeness has been studied and there are known separations under different complexity assumptions, see e.g.~the survey~\cite{Pavan03}.

In this short note, we answer an analogous question for arithmetic circuit complexity, the algebraic sibling of counting complexity. In arithmetic circuit complexity the most usual reduction notion are so-called \Sp-projections. Despite being  very restricted, \Sp-projections have been used to show nearly all of the completeness results in the area since the ground-breaking work of Valiant~\cite{valiant79}. It was only more recently that \Sc-reductions, a more permissive notion more similar to Turing- or oracle-reductions, have been defined in \cite{bur00} and used for some results (see e.g.~\cite{BriquelK09,Rugy-Altherre12,DurandMMRS14}).
Again the question comes up if there is a fundamental difference between these two notions of reductions. In fact, it was exactly this uncertainty about the relative power of \Sp-projections and \Sc-reductions that motivated the recent work Mahajan and Saurabh~\cite{MahajanS16}:
For the first time they prove a natural problem complete for the arithmetic circuit class $\vp$ under \Sp-projections,
where before there existed only such result under \Sc-reductions.

In this paper we answer the question of the relative strength of of \Sp-projections and \Sc-reductions:
We show unconditionally that over every field $\mathbb{F}$ there are explicit families of polynomials that are $\vnp$-complete over $\mathbb{F}$ under \Sc-reductions that are not $\vnp$-complete over $\mathbb{F}$ under \Sp-projections.
We also show that the question which polynomials are complete under which reductions depends on the underlying field in a rather subtle way.
It is a well known phenomenon that the permanent family,
which is $\vnp$-complete under \Sp-projections over fields of characteristic different from 2,
is contained in $\vp$ over fields of characteristic 2 and thus likely not $\vnp$-hard there.
We present a more subtle situation:
We give an explicit family of polynomials that is $\vnp$-complete under \Sc-reductions over all fields with more than 2 elements
and that is even $\vnp$-complete under \Sp-projections over a large class of fields including the complex numbers,
but over the real numbers it is only $\vnp$-complete under \Sc-reductions and not under \Sp-projections.

\paragraph*{Acknowledgements.} The authors would like to thank Dennis Amelunxen for helful discussions. Some of the research leading to this article was performed while the authors were at the Department of Mathematics at the University of Paderborn and at Texas A\&M University.

\section{Preliminaries}
We only give some very minimal notions of arithmetic circuit complexity. For more details we refer the reader to the very accessible recent survey~\cite{Mahajan14}. 

The basic objects to be computed in arithmetic circuit complexity are polynomials. More precisely, one considers so called \Sp-families of polynomials, which are sequences $(f_n)$ of multivariate polynomials such that the number of variables in $f_n$ and the degree of $f_n$ are both bounded by a polynomial in $n$.
We assume that each \Sp-family computes polynomials over a field $\mathbb{F}$ which will vary in this paper but is fixed for each \Sp-family.

A polynomial $f$ in the variables $X_1, \ldots, X_n$ is a projection of a polynomial $g$, in symbols $f\le g$, if $f(X_1, \ldots, X_n) = g(a_1, \ldots, a_m)$ where the $a_i$ are taken from $\{X_1, \ldots, X_n\} \cup \mathbb{F}$. The first reduction notion we consider in this paper are so-called \Sp-projections: A \Sp-family $(f_n)$ is a \Sp-projection of another \Sp-family $(g_n)$, symbol $(f_n) \le_\Sp (g_n)$ if there is a polynomially bounded function $t$ such that 
\[\exists n_0 \forall n \ge n_0 \colon f_n \le g_{t(n)}.\] 

Intuitively, \Sp-projections appear to be a very weak notion of reductions although surprisingly the bulk of completeness results in arithmetic circuit complexity can be shown with them. For some \Sp-families, though, showing hardness with \Sp-projections appears to be hard, and consequently, a more permissive reduction notion called \Sc-reductions has also been used.

The oracle complexity $L^{g}(f)$ of a polynomial $f$ with oracle $g$ is the minimum number of arithmetic operations $+$, $-$, $\times$, and evaluations of $g$ at previously computed values that are sufficient to compute~$f$ from the variables $X_1, X_2, \ldots$ and constants in $\mathbb{F}$.

Let $(f_n)$ and $(g_n)$ be \Sp-families of polynomials. We call $(f_n)$ a \emph{\Sc-reduction} of $(g_n)$, symbol $(f_n)\le_\Sc (g_n)$, if and only if there is a polynomially bounded function $t:\mathbb{N}\rightarrow \mathbb{N}$ such that the map $n\mapsto L^{g_{t(n)}}(f_n)$ is polynomially bounded.

Intuitively, if $(f_n) \le_\Sc (g_n)$, then we can compute the polynomial in $(f_n)$ with a polynomial number of arithmetic operations and oracle calls
to $g_{t(n)}$, where $t(n)$ is polynomially bounded.

Let $C_n$ denote the group of cyclic cyclic permutations on $n$ symbols and define the $n$th Hamiltonian cycle polynomial $\HC_n$ as $\HC_n := \sum_{\pi \in C_n} \prod_{i=1}^n X_{i,\pi(i)}$.

To keep these preliminaries lightweight, we omit the usual definition of $\vnp$ and instead define $\vnp$ to consist of all \Sp-families $(g_n)$ with $(g_n) \le_\Sp (\HC_n)$.

A \Sp-family $(g_n)$ that satisfies $(f_n)\le_\Sp (g_n)$ for all $f_n \in \vnp$ is called
$\vnp$-hard under \Sp-projections or $\vnp$-\Sp-hard for short.
Analogously,
a \Sp-family $(g_n)$ that satisfies $(f_n)\le_\Sc (g_n)$ for all $f_n \in \vnp$ is called
$\vnp$-hard under \Sc-reductions or $\vnp$-\Sc-hard for short.
If $(g_n)$ is $\vnp$-\Sp-hard and contained in $\vnp$, then $(g_n)$ is call $\vnp$-\Sp-complete.
Analogously for $\vnp$-\Sc-completeness.
Clearly if a family is $\vnp$-\Sp-complete, then it is also $\vnp$-\Sc-complete.

Note that a \Sp-family $(g_n)$ is $\vnp$-\Sp-hard (resp.~$\vnp$-\Sc-hard) iff $(\HC_n)\le_\Sp (g_n)$ (resp.~$(\HC_n)\le_\Sc (g_n)$).

\section{\Sc-reductions are strictly stronger than \Sp-projections}

In this section, we show that there are polynomials that are $\vnp$-\Sc-complete but not $\vnp$-\Sp-complete.
Let $X$ denote a new variable, unused by $\HC_n$ for any $n$.
Define
\[
P_n := X \cdot \HC_n + (\HC_n)^2.
\]

Note that $P_n$ is defined for every field.
We remark that $(P_n)$ can easily be shown to be contained in $\vnp$, because $\HC_n \in \vnp$ and the class $\vnp$ is closed under multiplication and addition~\cite{Valiant82} (see also~\cite[Theorem 2.19]{bur00}).

\begin{lemma}\label{lem:ccomplete}
$(P_n)$ is $\vnp$-\Sc-complete over every field.
\end{lemma}
\begin{proof}
Fix a field $\mathbb{F}$.
For field elements $\alpha \in \mathbb{F}$
let $P_n(X \leftarrow \alpha)$ denote $P_n$ with variable $X$ set to $\alpha$.
We observe that
\[
P_n(X \leftarrow 1) - P_n(X \leftarrow 0) = \HC_n
\]
and thus
$P_n$ is $\vnp$-\Sc-complete.
\end{proof}

\begin{lemma}\label{lem:notcomplete}
 $(P_n)$ is not $\vnp$-\Sp-complete over any field.
\end{lemma}
\begin{proof}
Let $f$ be any univariate polynomial in some variable $Y$ and let $f$ be of odd degree at least 3. We show that $f$ is not a projection of~$P_n$ for any $n$,
which finishes the proof because then the constant \Sp-family $(f)$ is not a \Sp-projection of $(P_n)$.

For a multivariate polynomial $h$ let $\deg_Y(h)$ denote the $Y$-degree of $h$,
which is the degree of $h$ interpreted as a univariate polynomial in $Y$ over the polynomial ring with additional variables.
Let $A$ be an $n \times n$ matrix whose entries are variables and constants.
We denote by $P_n(A)$ the linear projection of $P_n$ given by $A$.
We now analyze $\deg_Y(P_n(A))$.
Clearly $\deg_Y(X(A))\leq 1$.
If $\deg_Y(\HC_n(A)) \leq 1$, then $\deg_Y(P_n(A))\leq 2 < 3 \leq \deg_Y(f)$ and thus $P_n(A)\neq f$.
If $\deg_Y(\HC_n(A)) \geq 2$, then $\deg_Y(P_n(A)) = \deg_Y ((\HC_n(A))^2) = 2 \deg_Y (\HC_n(A))$.
But $\deg_Y(f)$ is an odd number, so in this case we also have $P_n(A)\neq f$.
\end{proof}

As a corollary we get that  \Sc-reductions yield strictly more complete problems that \Sp-projections.

\begin{theorem}\label{thm:main}
 For every field $\mathbb{F}$, $(P_n)$ is $\vnp$-\Sc-complete over $\mathbb{F}$, but not $\vnp$-\Sp-complete over $\mathbb{F}$.
\end{theorem}

\section{The dependence on the field}\label{sct:dependence}
In this section we construct a family $(Q_n)$ of polynomials that is
is $\vnp$-\Sc-complete over all fields with more than two elements,
but over the real numbers $(Q_n)$ is not $\vnp$-\Sp-complete.
This shows that the relative power of different reductions notions depends on the field and is thus likely quite complicated to characterize in general.

We consider the polynomials $Q_n$ defined on the matrix $(X_{ij})_{i,j\in [n]}$ defined by \[Q_n :=\sum_{\pi \in C_n} \prod_{i\in [n]} X_{i,\pi(i)} + \sum_{\pi \in C_n} \prod_{i\in [n]} X_{i,\pi(i)}^2.\] Note that $Q_n$ is similar to the polynomial $P_n$ considered before.
But unlike $P_n$ the homogeneous part of degree $n^2$ of $Q_n$ is \emph{not} $(\HC_n)^2$ but only contains a subset of the monomials.

Using Valiant's criterion, it is easy to see that $(Q_n) \in \vnp$, see for example \cite{bur00}[Proposition 2.20].

Although from its algebraic properties $Q_n$ might look very different from $P_n$,
the following Lemma can be proved exactly as Lemma~\ref{lem:ccomplete}.

\begin{lemma}\label{lem:ccompleteQ}
$(Q_n)$ is $\vnp$-\Sc-complete over every field with more than 2 elements.
\end{lemma}
\begin{proof}
 Fix a field $\FF$ with more than 2 elements.
 The proof is a simple interpolation argument.
 Choose $a \in \FF$ with $a \notin\{0,1\}$.
For a variable matrix
\[
X=\begin{pmatrix}
                         X_{1,1} & X_{1,2} & \cdots & X_{1,n}\\
                         X_{2,1} & X_{2,2} & \cdots & X_{2,n}\\
                         \vdots & \vdots & \ddots & \vdots \\
                         X_{n,1} & X_{n,2} & \cdots & X_{n,n}
                         \end{pmatrix}
\]
let $\bar X$ denote $X$ with the first row scaled by $a$:
\[
\bar X=\begin{pmatrix}
                         a X_{1,1} & a X_{1,2} & \cdots & a X_{1,n}\\
                         X_{2,1} & X_{2,2} & \cdots & X_{2,n}\\
                         \vdots & \vdots & \ddots & \vdots \\
                         X_{n,1} & X_{n,2} & \cdots & X_{n,n}
                         \end{pmatrix}.
\]
Clearly $\HC_n(\bar X) = a \HC_n(X)$.
Moreover,
\[
Q_n(\bar X) = a \HC_n(X) + a^2 \sum_{\pi \in C_n} \prod_{i\in [n]} X_{i,\pi(i)}^2.
\]
Therefore
\[
(a-a^2)\HC_n(X) = Q_n(\bar X) - a^2 Q_n(X).
\]
But $a-a^2 = a(1-a) \neq 0$ because $a \notin \{0,1\}$.
We conclude
\[
\HC_n(X) = \frac{1}{a-a^2} Q_n(\bar X) - \frac{a^2}{a-a^2}Q_n(X).
\]

It follows that $Q_n$ is even $\vnp$-\Sc-complete under linear $p$-projections, a restricted form of $c$-reductions (see \cite[p. 54]{bur00}).
\end{proof}

We now show that over the real numbers Lemma~\ref{lem:ccompleteQ} cannot be improved from \Sc-reductions to \Sp-projections.

\begin{lemma}\label{lem:notcompleteQ}
 $(Q_n)$ is not $\vnp$-\Sp-complete over $\mathbb{R}$.
\end{lemma}
\begin{proof}
 We show that the polynomial $X$ is not a projection of~$Q_n$ for any $n$. Assume this were not the case.
 Then there is an $(n\times n)$-matrix $A=(a_{ij})$ whose entries are variables or constants such that $P_n(A) = X$.
 W.l.o.g.\ we assume that no other variables than $X$ appear in $A$, so $a_{ij} \in \{X\}\cup \mathbb{R}$.
 Let $\sigma \in C_n$ be an $n$-cycle such that
$\prod_{i=1}^n a_{i\sigma(i)}$ has maximal degree. Obviously this degree is at least $1$. Then the monomial $\prod_{i=1}^n a_{i\sigma(i)}^2$ has at least degree $2$ and it cannot cancel out in $Q_n$ because
\begin{itemize}
 \item it cannot cancel with any $\prod_{i=1}^n a_{i\mu(i)}$ for an $n$-cycle $\mu$, because those all have smaller degrees, and
 \item it cannot cancel out with any $\prod_{i=1}^n a_{i\mu(i)}^2$, because those all have positive coefficients in $Q_n(A)$.
\end{itemize}

Thus $Q_n(A)$ has degree at least $2$, which implies that $Q_n(A)\neq X$.
\end{proof}

Interestingly, Lemma \ref{lem:notcompleteQ} does not generalize to arbitrary fields.

\begin{lemma}\label{lem:completeQ}
 Let $\mathbb{F}$ be a field such that there are elements $a_1, \ldots, a_s$ with $\sum_{i=1}^s a_i \ne 0$ and $\sum_{i=1}^s a_i^2 = 0$. Then $(Q_n)$ is $\vnp$-\Sp-complete over $\mathbb{F}$.
\end{lemma}
\begin{proof}
For an $(n\times n)$-matrix $A$ let $\HC(A)$ be the Hamiltonian cycle polynomial evaluated at $A$ and set $\HC(A^{(2)}) := \sum_{\sigma\in C_n} \prod_{i=1}^n a_{i\sigma(i)}^2$. With this notation clearly $Q_n(A) = \HC(A) + \HC(A^{(2)})$.
>From an $(s\times s)$-matrix $A$ and a $(t\times t)$-matrix $B$ we construct the $(s+t+2)\times(s+t+2)$ \emph{Hamiltonian connection matrix} $\con(A,B)$
as follows.
Let $G_A$ be the labeled digraph with adjacency matrix $A$ and let $G_B$ be the labeled digraph with adjacency matrix $B$.
The vertex corresponding to the first row and column in $A$ is called $v_A$, analogously for $v_B$.
The labeled digraph $G_A'$ is defined by replacing $v_A$ in $G_A$ by two vertices $v_A^{\text{in}}$ and $v_A^{\text{out}}$
such that the edges going into $v_A$ now go into $v_A^{\text{in}}$
and the edges coming out of $v_A$ now come out of $v_A^{\text{out}}$.
This operation increases the total number of vertices by one: $|V(G_A)|+1=|V(G_A')|$.

We create a labeled digraph $G_{\con(A,B)}$ as the union of $G_A'$ and $G_B'$ with two additional edges,
one going from $v_A^{\text{in}}$ to $v_B^{\text{out}}$ and the other from
$v_B^{\text{in}}$ to $v_A^{\text{out}}$, both labelled with 1.
Let $\con(A,B)$ denote the $(s+t+2)\times(s+t+2)$ adjacency matrix of $G_{\con(A,B)}$.

By construction we have a bijection between the set of Hamiltonian cycles in $G_{\con(A,B)}$ and
the set of pairs $(c_A,c_B)$ of Hamiltonian cycles $c_A$ in $G_A$ and $c_B$ in $G_B$.
Thus $\HC(\con(A,B))=\HC(A)\HC(B)$ and $\HC(\con(A,B)^{(2)})=\HC(A^{(2)})\HC(B^{(2)})$.
Therefore
\begin{equation}\label{eqn:decompose}
Q_{s+t+2}(\con(A,B)) = \HC(A)\HC(B) + \HC(A^{(2)})\HC(B^{(2)}).
\end{equation}

Let $a:= \sum_{i=1}^s a_i$ and  \[A:= \begin{pmatrix} 0 & a^{-1} & a^{-1} & \ldots & a^{-1} & a^{-1}\\
a_1 & 0 & 0 & \ldots & 0 & 1 \\
a_2 & 1 & 0 & \ldots & 0 & 0 \\
a_3 & 0 & 1 & \ldots & 0 & 0 \\
\vdots & \vdots & \vdots & \ddots & \vdots & \vdots \\
a_s & 0 & 0 & \ldots & 1 & 0 \\
\end{pmatrix}
\]
It is easy to verify that $\HC(A) = \sum_{i=1}^s a_i a^{-1} = 1$ and $\HC(A^2) = \sum_{i=1}^s a_i^2 (a^{-1})^2 = \left(\sum_{i=1}^s a_i^2\right) a^{-2} = 0$.

Thus we get with (\ref{eqn:decompose})
\[Q_{s+t+2}(\con(A,B)) = \HC(A)\HC(B) + \HC(A^{(2)})\HC(B^{(2)}) = \HC(B) \]
for every $(t\times t)$-matrix $B$.

Thus the Hamiltonian cycle family $(\HC_n)$ is a \Sp-projection of $(Q_n)$ and the claim follows.
\end{proof}

\begin{corollary}\label{cor:complete2Q}
 \begin{enumerate}
  \item[a)] $(Q_n)$ is $\vnp$-\Sp-complete over $\mathbb{C}$.
  \item[b)] $(Q_n)$ is $\vnp$-\Sp-complete over any field of characteristic greater than $2$.
 \end{enumerate}
\end{corollary}
\begin{proof}
 a) Set $s:=2$ and $a_1=1$ and $a_2 =i$. We have $a_1 + a_2 =1+i \neq 0$ and $a_1^2 + a_2^2 = 0$ and thus the claim follows by Lemma \ref{lem:completeQ}.
 
b) Let $p>2$ be the characteristic of the field and set $s:=p$. We have \[\sum_{i=1}^{\frac{p-1}{2}} 1 + \sum_{i=1}^{\frac{p+1}{2}} (-1) = -1 \ne 0\] and \[\sum_{i=1}^{\frac{p-1}{2}} 1^2 + \sum_{i=1}^{\frac{p+1}{2}} (-1)^2 = p\cdot 1 = 0.\] With Lemma \ref{lem:completeQ} the claim follows.
\end{proof}

\section{Conclusion}

We have shown that for all fields \Sc-reductions and \Sp-projections differ in power. 
Note that one could show versions of Theorem~\ref{thm:main} for essentially all other complexity classes from arithmetic circuit complexity, as long as they contain complete families of homogeneous polynomials and the polynomial~$X$. Since the proofs are essentially identical, we have not shown these results here.

We have also shown that the question which families are complete under which reductions also depends on the field. This indicates that understanding the exact power of different reduction notions is probably very complicated.

Another question is with respect to the naturalness of our separating examples. They have been specifically designed for our results and apart from that we do not consider them very interesting. Can one show that the more natural polynomials in \cite{BriquelK09,Rugy-Altherre12,DurandMMRS14} which were shown to be complete under \Sc-reductions are not complete under \Sp-projections?

\newcommand{\etalchar}[1]{$^{#1}$}

\end{document}